\documentclass[draftcls,12pt,onecolumn,oneside]{IEEEtran}
\usepackage{times,amsmath,amssymb,graphicx,epsfig,cite}
\usepackage{subfigure,url}

 \DeclareMathOperator*{\ineq}{\leq}

\newcommand{\less}{\le}

\newcommand{\links} {{\mathcal E}}
\newcommand{\sink}{d}

\newcommand{\source}{s}
\newcommand{\incoming}{{\mathcal I}}
\newcommand{\outgoing}{{\mathcal O}}
\newtheorem{lemma}{Lemma}
\newtheorem{theorem}{Theorem}

\newtheorem{claim}{Claim}

\newcommand{\cX} {{\cal X}}

\renewcommand{\baselinestretch}{1.45}
\begin{document}

\title{On Secure Network Coding with Nonuniform or Restricted Wiretap Sets
  \thanks{This work has been supported in part by subcontract \#069144 issued
    by BAE Systems National Security Solutions, Inc. and supported by the
    Defense Advanced Research Projects Agency (DARPA) and the Space and Naval
    Warfare System Center (SPAWARSYSCEN), San Diego under Contracts No.
    N66001-08-C-2013 and W911NF-07-1-0029, by NSF grants CNS 0905615, CCF
    0830666, CCF 1017632, and by Caltech's Lee Center for Advanced Networking.
    The material of this paper was presented in part at the 2010 Information
    Theory and Applications Workshop, La Jolla, CA, and the 2010 IEEE
    Information Theory Workshop, Dublin, Ireland.}\vspace{-2mm}} \author{Tao
  Cui, Tracey Ho, \IEEEmembership{ Member,~IEEE}, and J{\"o}rg Kliewer,
  \IEEEmembership{ Senior Member,~IEEE} \thanks{Tao Cui and Tracey Ho are with
    the Department of Electrical Engineering, California Institute of
    Technology, Pasadena, CA 91125, USA (Email: \{taocui,tho\}@caltech.edu).}
  \thanks{J{\"o}rg Kliewer is with the Klipsch School of Electrical and Computer
    Engineering, New Mexico State University, Las Cruces, NM 88003, USA
    (Email: jkliewer@nmsu.edu).}\vspace{-5mm}}
\maketitle

\begin{abstract}
The secrecy capacity of a network, for a given collection of permissible wiretap sets, is the maximum rate of communication such that observing links in any permissible wiretap set reveals no information about the message. This paper considers secure network coding with nonuniform or restricted wiretap sets, for example, networks with unequal link capacities where a wiretapper can wiretap any subset of $k$ links, or networks where only a subset of links can be wiretapped. Existing
results show that for the case of uniform wiretap sets (networks with equal capacity links/packets where any $k$ can be wiretapped), the
secrecy capacity is given by the cut-set bound, and can be achieved by injecting $k$
random keys at the source which are decoded at the sink along with
the message. This is the case whether or not the communicating users have information about the choice of wiretap set. In contrast, we show that for the nonuniform case, the cut-set bound is not achievable in general when the wiretap set is unknown, whereas it is achievable when the wiretap set is made known.
We give achievable
strategies where random keys are canceled at intermediate non-sink
nodes, or injected at intermediate non-source nodes. Finally, we show that
determining the secrecy capacity is a NP-hard problem.

\end{abstract}

\begin{keywords}\vspace{-4mm}
Secrecy capacity, network coding,
information-theoretic security, cut-set bound, network interdiction, NP-hard.

\end{keywords}

\newpage

\section{Introduction}
Information-theoretically secure communication uses coding to ensure that an adversary that wiretaps a subset of network links obtains no information about the secure message. The secrecy capacity of a network, for a given collection of permissible wiretap sets, is defined as the maximum rate of communication such that any one of the permissible wiretap sets reveals no information about the message. In general, the choice of wiretap set is unknown to the communicating users, though we also discuss the case of known wiretap set where the encoding and decoding functions are allowed to depend on the choice of wiretap set, in which case the secrecy capacity is the maximum rate achievable under the worst case wiretap set.
%Information-theoretic security is a principle that can strengthen
%the security of wireless networks at the physical layer by using
%coding to guarantee that the messages sent cannot be decoded by a
%malicious eavesdropper.

A theoretical basis for information-theoretic security was given in the seminal
paper by Wyner \cite{wyner75} using Shannon's notion of perfect
secrecy \cite{shannon48}, where a coset coding scheme based on a
linear maximum distance separable code was used to achieve
security for a wiretap channel. More recently, information-theoretic
security has been studied in networks with general topologies. The secure network coding problem, where a wiretapper observes an unknown set of links,
was introduced by Cai and Yeung \cite{cai02}. They proposed a  coding strategy, which we refer to as the global key strategy, in which the source injects random key symbols that are decoded at the sink along with the message. They  showed achievability of this strategy in the nonuniform case where a wiretapper can observe one of an arbitrary given collection of wiretap link sets, and optimality of this strategy for multicast in the uniform case where each link has equal capacity and a wiretapper can observe up to  $k$ links.  For the uniform case, various
constructions of secure linear network
codes have been proposed in e.g.~\cite{feldman04,rouayheb07wiretap}. %For this network
%model, trade-offs between security, code alphabet size, and
%multicast rate of secure linear network codes are considered in
%\cite{feldman04}.
Other related work on secure network communication includes weakly secure codes~\cite{bhattad05weakly} and wireless erasure networks \cite{mills08}.

In this paper, we consider secure communication over wireline networks in the nonuniform case.
In the case of throughput optimization
without security requirements, the assumption that all links have
unit capacity is made without loss of generality, since links of
larger capacity can be modeled as multiple unit capacity links in
parallel.  However, in the secure communication problem, such an
assumption cannot be made without loss of generality.  Indeed, we
show in this paper that there are significant differences between the uniform and nonuniform cases. For the case of uniform wiretap sets, the multicast secrecy capacity is given by the cut set bound,
whether or not the choice of wiretap set is known, and is achieved by the global key strategy \cite{cai02}. In contrast, the nonuniform case is more complicated, even for a single source and sink.  We
show that the secrecy capacity is
not the same in general when the location of the wiretapped links is
known or unknown. We give new achievable strategies where random
keys are canceled at intermediate non-sink nodes or injected at
intermediate non-source nodes, and show that these strategies can outperform the global key strategy. Finally, we show that determining the
secrecy capacity is an NP-hard problem.

\section{Network Model and Problem Formulation}\label{sect5_2}

In this paper we focus on acyclic graphs for simplicity; we expect
that our results can be generalized to cyclic networks using the
approach in \cite{koetter01,ho03} of working over fields of rational
functions in an indeterminate delay variable.

We model a wireline
network by a directed acyclic graph
$\mathcal{G}=(\mathcal{V},\mathcal{E})$, where $\mathcal{V}$ is the
vertex set and $\mathcal{E}$ is the directed link set. Each link
$(i,j)\in \mathcal{E}$ is a noise-free bit-pipe with a given capacity $c_{i,j}$.
We denote the set of incoming links $(w,v)$ of a node $v$ by $\incoming(v)$ and
the set of outgoing links $(v,w)$ of $v$ by $\outgoing(v)$.

A source node $s\in\mathcal{V}$ observes
a random source process $X_\source$
taking values from a discrete alphabet $\cX_\source$.
A sink node $d\in\mathcal{V}$ wishes to reconstruct $X_\source$ with probability of error going to zero with the coding blocklength.

An eavesdropper can wiretap a set $\mathcal{A}$ of links chosen from a known collection $\mathcal{W}$ of possible wiretap sets. Without loss of generality we can restrict our attention to maximal wiretap sets, i.e.~no set in $\mathcal{W}$ is a subset of another. The choice of wiretap set $\mathcal{A}$ is unknown to the communicating nodes, except where otherwise specified in this paper. In the case of known wiretap set,  the wiretapper can choose an arbitrary wiretap set $\mathcal{A}$ in $\mathcal{W}$  which is then revealed to the communicating nodes.

A block code of blocklength $n$ is defined by
a mapping
$$f_{e}^{(n)}:\cX_\source^n\rightarrow\{1,\dots, 2^{nc_{e}}\},
\;e\in\outgoing(\source) $$
from $X_\source^n$ to the vector
transmitted on each outgoing link $e$
of the source  $\source$,
a mapping
$$f_{e}^{(n)}:\prod_{d\in\incoming(v)}\{1,\dots, 2^{nc_{d}}\}
\rightarrow\{1,\dots, 2^{nc_{e}}\},\;e\in\outgoing(v) $$
from the vectors received by  a non-source node $v$ to the vectors transmitted
on each outgoing link $e$ of $v$,
and a mapping
$$g_\sink^{(n)}:\prod_{d\in\incoming(\sink)}\{1,\dots,2^{nc_{d}}\}
\rightarrow\cX_\source^n $$
from  the vectors received by the sink $\sink$ to the decoded output.
Node mappings are applied in topological order;
each node receives input vectors
from all its incoming links
before applying the mappings corresponding to its outgoing links.

The secrecy
capacity is defined as the highest possible source-sink communication rate for which there exists a sequence of block codes such that the probability of decoding error at the sink goes to zero and, for any choice of $\mathcal{A}\in\mathcal{W}$, the
message communicated is information theoretically secret, i.e.~has zero
mutual information with the wiretapper's observations.

In Section~\ref{section:general} we give a cut set bound and achievable strategies for this general problem.
In  Sections \ref{section:unachievable} and \ref{section:NP}, we show that the cut set bound is unachievable and that finding the secrecy capacity is NP hard, even for the following special cases:
%In the following, we consider two related scenarios:
\begin{enumerate}
\item Scenario 1 is a wireline network with \emph{equal} link capacities,
  where the wiretapper can wiretap an unknown subset of $k$ links from
  a known collection of vulnerable network links. %The objective is to
%  maximize the communication rate subject to the requirement that the
%  message is secret regardless of the choice of wiretapped links.
\item Scenario 2 is a wireline network with \emph{unequal} link capacities,
  where the wiretapper can wiretap an unknown subset of $k$ links from
  the entire network.  %Again, the objective is to maximize the communication
%  rate subject to the requirement that the message is secret
%  regardless of the choice of wiretapped links.
%\item Scenario 2 is a wireline network with equal link capacities,
%  where the wiretapper can wiretap an unknown subset of $k$ links from
%  a known collection of vulnerable network links. The objective is to
%  maximize the communication rate subject to the requirement that the
%  message is secret regardless of the choice of wiretapped links.
%  Scenario 2 is used as a bridge for studying scenario 1. In the
%  following sections, scenario 2 is discussed first. We then convert a
%  network considered in scenario 2 to a corresponding network for
%  scenario 1 such that the same result holds.
%  \item The network interdiction problem \cite{wood93} is to
%minimize the maximum flow of the network when $k$ links are removed
%from the network. This is equivalent to a secure network coding problem where the wiretapper's location is made known to the communicating users.
\end{enumerate}
It is convenient to show these results for  Scenario 1 first, and then
show the corresponding results for Scenario 2,  by converting the
Scenario 1 networks considered into corresponding Scenario 2 networks
for which the same result holds.

Although, for the sake of simplicity, we focus on
single-source single-sink networks, the cut-set
bound and strategy 2 in Section~III can be easily extended to
multicast networks, whereas the results discussed in Sections~IV and~V
directly apply to both multicast and non-multicast cases since the
single-source single-sink case represents a special case for both.

\section{Cut-Set Bound and Achievable Strategies}
\label{section:general}
In this section, we consider the general wireline problem with unequal link capacities where the eavesdropper can wiretap an unknown set $\mathcal{A}$ of links chosen from a known collection $\mathcal{W}$ of possible wiretap sets. We state a cut-set upper bound on capacity, and give two new achievable strategies and examples in which they outperform the existing global key strategy. Using the combined intuition from these examples, we show in Section~\ref{section:unachievable} that the cut-set bound is unachievable in general.  One of the achievable strategies is used in Section~\ref{section:NP} to show that finding the secrecy capacity in general is NP-hard.

\subsection{Cut-Set Bound}

%For each node $i\in \mathcal{V}$, $\mathcal{N}_{\mathcal{O}}(i)$ and
%$\mathcal{N}_{\mathcal{I}}(i)$ denote the set of in-neighbors and
%out-neighbors of $i$, i.e.,
%\begin{equation}\label{eq5_001}
%\mathcal{N}_{\mathcal{I}}(i)=\left\{j|(j,i)\in
%\mathcal{E}\right\},\quad
%\mathcal{N}_{\mathcal{O}}(i)=\left\{j|(i,j)\in \mathcal{E}\right\}.
%\end{equation}
Let $\mathcal{S}^c$ denote the set complement of a set $\mathcal{S}$. A cut for $x,y\in \mathcal{V}$ is a partition of $ \mathcal{V}$ into
two sets $ \mathcal{V}_x$ and $ \mathcal{V}_x^c$ such
that $x\in \mathcal{V}_x$ and $y\in \mathcal{V}_x^c$. For the $x-y$
cut given by $ \mathcal{V}_x$, the cut-set
$[\mathcal{V}_x,\mathcal{V}_x^c]$ is the set of links going from
$\mathcal{V}_x$ to $\mathcal{V}_x^c$, i.e.,
\begin{equation}\label{eq5_002}
[\mathcal{V}_x,\mathcal{V}_x^c]=\left\{(u,v)|(u,v)\in
\mathcal{E},\,u\in \mathcal{V}_x,\,v\in \mathcal{V}_x^c\right\}.
\end{equation}

\begin{theorem}\label{th31}
Consider a network of point-to-point links, where link $(i,j)$ has capacity $c_{i,j}$. The secrecy capacity is upper bounded by
\begin{equation}\label{eq5_003}
\min_{\left\{\mathcal{V}_s:\,\mathcal{V}_s\text{ is an } s-d
\text{ cut}\right\}}\min_{\mathcal{A}\in\mathcal{W}}
\sum_{(i,j)\in
[\mathcal{V}_s,\mathcal{V}_s^c]\cap\mathcal{A}^c}c_{i,j}.
\end{equation}This bound applies whether or not the communicating nodes have knowledge of the chosen wiretap set $\mathcal{A}$.
%where
%\begin{equation}\label{eq5_004}
%H(M|\mathbf{Z})=H(M).
%\end{equation}
\end{theorem}

\begin{proof} Consider any source-sink cut $\mathcal{V}_s$ and any wiretap set
$\mathcal{A}\in\mathcal{W}$. Denote by
$\mathbf{X}$ the transmitted signals from nodes in $\mathcal{V}_s$
over links in $[\mathcal{V}_s,\mathcal{V}_s^c]$ and denote by
$\mathbf{Y}$ and $\mathbf{Z}$  the observed signals from links  in $[\mathcal{V}_s,\mathcal{V}_s^c]$ and in
$[\mathcal{V}_s,\mathcal{V}_s^c]\cap\mathcal{A}$
respectively. %Let $\mathbf{Y}$ contain the observations from links
%in $[\mathcal{V}_s,\mathcal{V}_s^c]$, and let $\mathbf{Z}$ be the wiretapped
%signals from $\mathcal{A}$.
We consider block coding with block
length $n$ and secret message rate $R_s$. By the perfect secrecy requirement
$H(M|\mathbf{Z}^n)=H(M)$  we have
\begin{equation}\label{eq5_018}
\begin{split}
nR_s\leq & H(M|\mathbf{Z}^n)\\
\ineq^{(a)} & H(M|\mathbf{Z}^n)-H(M|\mathbf{Y}^n)+n\epsilon_n\\
= & H(M|\mathbf{Z}^n)-H(M|\mathbf{Y}^n,\mathbf{Z}^n)+n\epsilon_n\\
= & I(M;\mathbf{Y}^n|\mathbf{Z}^n)+n\epsilon_n\\
\ineq^{(b)} & I(\mathbf{X}^n;\mathbf{Y}^n|\mathbf{Z}^n)+n\epsilon_n\\
= &H(\mathbf{X}^n|\mathbf{Z}^n)-H(\mathbf{X}^n|\mathbf{Y}^n,\mathbf{Z}^n)+n\epsilon_n\\
\less &H(\mathbf{X}^n|\mathbf{Z}^n)+n\epsilon_n\\
%\ineq^{(c)}&\sum_{i=1}^nH(Y_{i}|Z_{i})-\sum_{i=1}^nH(Y_{i}|X_i,Z_{i})+n\epsilon_n,\\
%=& nI(\mathbf{X};\mathbf{Y}|\mathbf{Z})
%+n\epsilon_n,\\
%=& n\left(H(\mathbf{X}|\mathbf{Z}
%)-H(\mathbf{X}|\mathbf{Z},\mathbf{Y})\right)+n\epsilon_n,\\
%=& n\left(H(\mathbf{X}|\mathbf{Z}
%)-H(\mathbf{X}|\mathbf{Y})\right)+n\epsilon_n,\\
%\leq& n\max_{p(\mathbf{X})}\left(I(\mathbf{X};\mathbf{Y}
%)-I(\mathbf{X};\mathbf{Z})\right)+n\epsilon_n\\
\less &n\sum_{(i,j)\in
[\mathcal{V}_s,\mathcal{V}_s^c]\cap\mathcal{A}^c}c_{i,j}+n\epsilon_n,
\end{split}
\end{equation}
where $\epsilon_n\rightarrow0$ as $n\rightarrow +\infty$ and
\newcounter{Lcount1}
  %    set the "default" label to print counter as a Roman numeral
  \begin{list}{$(\alph{Lcount1})$}
  %    inform the list command to use this counter
    {\usecounter{Lcount1}
  %    set rightmargin equal to leftmargin
    \setlength{\rightmargin}{\leftmargin}}
  %    we can now begin the "items"
  \item is due to Fano's inequality;

  \item is due to the data processing inequality and the fact that
    $M\rightarrow\mathbf{X}^n\rightarrow
    \mathbf{Y}^n\rightarrow\mathbf{Z}^n$ forms a Markov chain;

  %\item is due to the definition of the mutual information.
  \end{list}
\end{proof}
If the choice of wiretap set $\mathcal{A}$ is known to the communicating nodes, the cut-set bound in this case is also (\ref{eq5_003}), which is achievable using a network code that does not send any flow on links in $\mathcal{A}$. In contrast, we show in Section~\ref{section:unachievable} that the cut-set bound is not achievable in general when the wiretap set $\mathcal{A}$ is unknown.

\subsection{Achievable Strategies for Unknown Wiretap Set} \label{sect5_2_3}
%In the following we discuss scenario 1 from above.
In the case of unit link capacities, the secrecy capacity can be
achieved using global keys generated at the source and decoded at
the sink \cite{cai02}. The source transmits $R_s$ secret information
symbols and $R_w$ random key symbols, where $R_s+R_w$ is equal to the
min-cut of the network. This scheme does not achieve capacity in
general networks with unequal link capacities. Intuitively, this is
because the total rate of random keys is limited by the min cut from
the source to the sink, whereas more random keys may be required to
fully utilize large capacity cuts with large capacity links.

Capacity can be improved by using a combination of local and global
random keys. A local key is injected at a non-source node and/or
canceled at a non-sink node.  However, it is complicated to optimize
over all possible combinations of nodes at which keys are injected
and canceled.  Thus, we propose the following more tractable
constructions, which we will use to develop further results in subsequent sections.

%Let
%$\mathcal{W}$ be the set of all possible maximal subsets of links
%that the wiretapper may access simultaneously. For scenario 2,
%$\mathcal{W}=\{\mathcal{A}|\mathcal{A}\subseteq
%\mathcal{E},\,|\mathcal{A}|=k\}$.

\emph{Strategy 1: Random Keys Injected by Source and Possibly
Canceled at Intermediate Nodes}

\begin{figure}
\centering
  \includegraphics[scale=0.9]{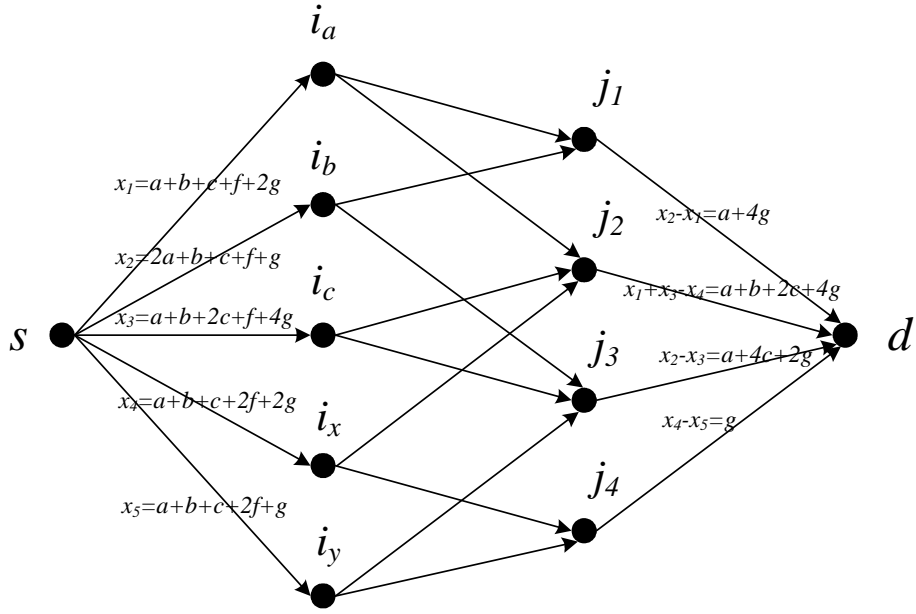}\\
  \caption{An example of Strategy 1, where any two of the five links in the first layer can be wiretapped.  Capacity 3 is achieved by canceling at the second layer one of the two random keys injected  by the source. The operation is on a finite field $GF(5)$.}\label{f3_1}
\end{figure}

Our first construction achieves secrecy with random keys injected only at the source. The source carries out pre-coding so that random keys are canceled at intermediate nodes and the sink receives the intended message without interference from the random keys. As such, it applies in the single-source, single-sink case, and is useful in networks where the incoming capacity of the sink is too small to accommodate the message plus all the keys needed in the network.
An example is given in
Fig.~\ref{f3_1}, where each link has unit capacity, the number of wiretapped links is $k=2$, and
only the first layer of the three layer network is allowed to be
wiretapped.
%Let $c$
%denote the minimum cut after deleting any $k$ links in the first
%layer of the graph. As the wiretapped links are connected to the
%source directly, the min-cut between each virtual sink and the
%source is at least $c+k$. Since $c$ is the cut-set upper bound on
%the secrecy rate, by using the key cancelation scheme the secrecy
%rate $c$ is achievable, which is equal to the secrecy rate when the
%location of wiretap links is known.
The secret message rate $R_s=3$ is achievable by using the strategy in Fig.~\ref{f3_1}, where the operation is on a finite field $GF(5)$. In Fig.~\ref{f3_1}, $a,b,c$ are secret messages and $f,g$ are keys. The message on the $i$-th link in the first layer is denoted as $x_i$, $i=1,2,3,4,5$. The key $f$ is canceled at the second layer and the key cancelation scheme is labeled on the last layer links. It is easy to see that $H(x_i,x_j|a,b,c)=2$, $\forall i\neq j$ which means perfect secrecy is achieved. At the same time, the sink $d$ can decode $a,b,c$ and the key $g$.
When key
cancelation is not applied, let $R_s$ and $R_w$ be the secrecy rate and
the random key rate at the source, respectively. Let $z$ be the
total rate of transmission on the first layer. To achieve secrecy, we must
have $R_w\geq \frac{2}{5}z$, where the cut-set condition on the first
layer requires $R_s+R_w\leq z$.  Since the sink needs to decode both
message and random key symbols from the source, the cut-set
condition on the last layer requires $R_s+R_w\leq 4$. Combining these we
obtain $R_s\leq \max_z\min(4-\frac{2}{5}z,\;\frac{3}{5}z) =\frac{12}{5}$, which is strictly less than 3.

To formally develop the Strategy 1 construction, we will use the following result:
\begin{claim}[{\cite[Corollary
19.21]{yeung08}}] Given an acyclic network, there exists, for a sufficiently large finite field, a linear network code in which the dimension of the received subspace at each non-source node $t$ is $\min (\omega,\mbox{maxflow} (t))$, where $\omega$ is the dimension of the message subspace.
\label{Raymond}
\end{claim}

\begin{figure}
\centering
  \includegraphics[scale=0.7]{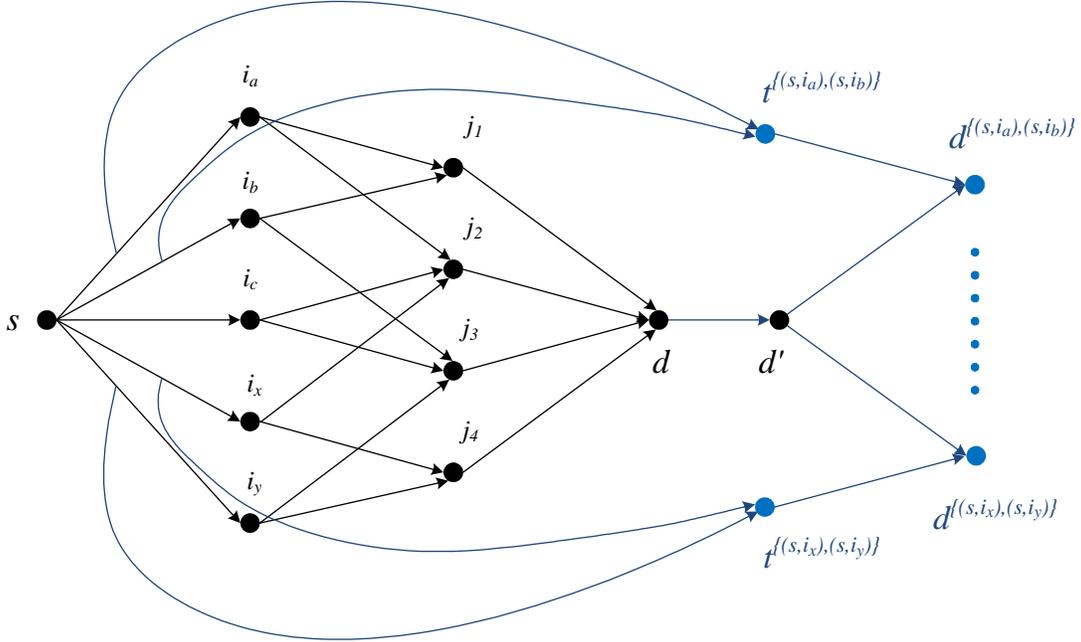}\\
  \caption{Illustration of Strategy 1, an achievable construction where random keys are injected by the source and
  possibly canceled at intermediate nodes. In this figure, $k=2$ and only the 5 links in the first layer can be wiretapped.}\label{f3}
\end{figure}
Let $R_s$ denote the secret message rate and $z_{i,j}$ the transmission rate on each network link $(i,j)\in\links$, whose values we will discuss how to choose below.  Consider the graph
$\mathcal{G}$ with the capacity of each link $(i,j)\in\links$ set as
$z_{i,j}\le c_{i,j}$.
As
illustrated in Fig.~\ref{f3}, augment the graph as follows:\begin{itemize}\item Connect each subset of links $\mathcal{A}\in \mathcal{W}$ to a virtual node $t^{\mathcal{A}}$: more precisely, for each directed link $(i,j)\in\links$ in the network, create a node $v_{i,j}$ and replace $(i,j)$ by two links $(i,v_{i,j})$ and $(v_{i,j},j)$ of capacity $z_{i,j}$, and for each $(i,j)\in\mathcal{A}$ create a link $(v_{i,j},t^{\mathcal{A}})$ of capacity $v_{i,j}$. Let $R_{s\rightarrow \mathcal{A}}$ be the max flow/min cut capacity between $s$ and $t^{\mathcal{A}}$. \item Add a virtual sink node $d^\prime$ and join the actual sink $d$ to $d^\prime$ by a link $(d,d^\prime)$ of
capacity of $R_s$.
\item Connect
both $t^{\mathcal{A}}$ and the virtual sink $d^\prime$ to a virtual sink
$d^{\mathcal{A}}$ by adding a link $(t^{\mathcal{A}},d^{\mathcal{A}})$ of  capacity
$R_{s\rightarrow \mathcal{A}}$  and a link $(d^\prime,d^{\mathcal{A}})$ of capacity $R_s$,  respectively.\end{itemize}

The source sends a secret
message $\mathbf{v}=[v_1,\ldots,v_{R_s}]^T$  along with
$R_w$ random key symbols $\mathbf{w}=[w_1,\ldots,w_{R_w}]^T$.\footnote{We
assume that $R_s$ and $R_{s\rightarrow \mathcal{A}}$ are integers,
which can be approximated arbitrarily closely by scaling the
capacity of all links by the same factor.} The values of $R_s$, $R_w$, and $z_{i,j}$ are chosen such
that each virtual sink $d^{\mathcal{A}}$ can decode
$R_s+R_{s\rightarrow \mathcal{A}}$ linear combinations of message and random key symbols, and the sink $\sink$ can decode the $R_s$ message
symbols.
Specifically, if for each $\mathcal{A}$ the rate $R_s+R_{s\rightarrow \mathcal{A}}$ equals the min-cut capacity
between the source and the virtual sink $d^{\mathcal{A}}$ and
$R_{s\rightarrow \mathcal{A}}\leq R_w$, by using Claim~\ref{Raymond}, there exists a network code such that each
$d^\mathcal{A}$ receives $R_s+R_{s\rightarrow \mathcal{A}}$ linearly
independent combinations of $\mathbf{v}$ and $\mathbf{w}$ when the
finite field size is sufficiently large ($q>{|\mathcal{E}|
\choose k}$). Let the signals received at
a particular virtual sink $d^\mathcal{B}$ be denoted as
$\mathbf{M}_{\mathcal{B}}[\mathbf{v}^T,\mathbf{w}^T]^T$, where
$\mathbf{M}_{\mathcal{B}}$ is an $(R_s+R_{s\rightarrow \mathcal{B}})\times (
R_s+R_w)$ received coding matrix with full row rank. We can add
$R_w-R_{s\rightarrow \mathcal{B}}$ rows to
$\mathbf{M}_{\mathcal{B}}$ to get a full rank
$(R_s+R_w)\times(R_s+R_w)$ square matrix
$\tilde{\mathbf{M}}_{\mathcal{B}}$. We then precode the secret
message and keys using $\tilde{\mathbf{M}}_{\mathcal{B}}^{-1}$,
i.e., the source transmits
$\tilde{\mathbf{M}}_{\mathcal{B}}^{-1}[\mathbf{v}^T,\mathbf{w}^T]^T$, so that link $(d^\prime,d^{\mathcal{B}})$ carries $\mathbf{v}$.

\begin{claim}Strategy 1 allows the sink to decode the message $\mathbf{v}$ and achieves perfect secrecy.\end{claim}

\begin{proof}Since $(d,d^\prime)$ is the only incoming link of $(d^\prime,d^{\mathcal{B}})$, and both links have capacity $R_s$ which is equal to the rate of the message $\mathbf{v}$, link  $(d,d^\prime)$ carries exactly $\mathbf{v}$. This implies that sink $d$ receives $\mathbf{v}$. Furthermore,
for any virtual sink $d^\mathcal{A}$, the received coding matrix
with precoding is
$\mathbf{M}_{\mathcal{A}}\tilde{\mathbf{M}}_{\mathcal{B}}^{-1}$,
which is a full row rank matrix. As
$\mathbf{M}_{\mathcal{A}}\tilde{\mathbf{M}}_{\mathcal{B}}^{-1}$ is a
full row rank matrix, the coding vectors of the received signals from
the set $\mathcal{A}$ of wiretapping links span a rank
$R_{s\rightarrow \mathcal{A}}$ subspace that is linearly independent
of the set of coding vectors corresponding to message $\mathbf{v}$ received on
$(d^\prime,d^{\mathcal{A}})$. Therefore, the signals received on $\mathcal{A}$ are independent of the message $\mathbf{v}$, and perfect secrecy is achieved.\end{proof}

Note that applying
$\tilde{\mathbf{M}}_{\mathcal{B}}^{-1}$ causes the random keys injected by
the source to be either canceled at intermediate nodes or
decoded by the sink.

%Unlike the uniform wiretap problem where it is optimal to transmit at the full capacity of each link, in the non-uniform case it may be useful to transmit on a link $(i,j)$ at a rate $z_{i,j}$ less than the available link capacity $c_{i,j}$.
It remains to optimize over values of $R_s$, $R_w$ and $z_{i,j}$ such
that  for each $\mathcal{A}$ the rate $R_s+R_{s\rightarrow \mathcal{A}}$ equals the min-cut capacity
between $s$ and $d^{\mathcal{A}}$ and
$R_{s\rightarrow \mathcal{A}}\leq R_w$. Since computing $R_{s\rightarrow \mathcal{A}}$ (the min-cut capacity between $s$ and $t^{\mathcal{A}}$) for arbitrary $z_{i,j}$ involves a separate max flow computation, to simplify the optimization, we can constrain  $R_{s\rightarrow \mathcal{A}}$ to be equal to some upper bound $U_\mathcal{A}$ on  $R_{s\rightarrow \mathcal{A}}$, and thereby obtain an achievable
secrecy rate using Strategy 1. For instance, we can take $U_\mathcal{A}$ to be $\sum_{(i,j)\in
\mathcal{A}}z_{i,j}$, or alternatively take $U_\mathcal{A}$ to be the min-cut capacity between $s$ and $t^{\mathcal{A}}$ on the graph with the original link capacities $c_{i,j}$. We can write a linear program (LP)  for
this key cancelation strategy as follows:%
\begin{equation}\label{eq321_1}
\begin{split}
\max \text{ }&R_s\\
\text{subject to } &\sum_{(i,j)\in
\mathcal{E}}f_{i,j}^{\mathcal{A}}-\sum_{(i,j)\in
\mathcal{E}}f_{j,i}^{\mathcal{A}}=\left\{
                                    \begin{array}{cc}
                                      R_s+U_\mathcal{A}, & \text{if } i=s, \\
                                      -R_s-U_\mathcal{A}, & \text{if } i=d^{\mathcal{A}}, \\
                                      0, & \text{otherwise}, \\
                                    \end{array}
                                  \right.\\
& \hspace{85mm}\forall \mathcal{A}\in \mathcal{W},\\
& f_{i,j}^{\mathcal{A}}\leq z_{i,j}\leq c_{i,j},\, \forall (i,j)\in
\mathcal{E},
\end{split}
\end{equation}
where $f_{i,j}^{\mathcal{A}}$ represents the rate of flow on link $(i,j)$
intended for the virtual sink $d^\mathcal{A}$. The conditions on conservation of flow $f_{i,j}^{\mathcal{A}}$ ensure that the min cut between the source and $d^\mathcal{A}$ is at least $R_s+U_\mathcal{A}$. Since the only incoming links of $d^\mathcal{A}$ are $(t^{\mathcal{A}},d^{\mathcal{A}})$ of  capacity
$R_{s\rightarrow \mathcal{A}}$  and  $(d^\prime,d^{\mathcal{A}})$ of capacity $R_s$, this implies that $R_{s\rightarrow \mathcal{A}}$ equals the upper bound $U_\mathcal{A}$.  Thus, the optimal value of (\ref{eq321_1}) gives an
achievable secrecy rate.

\emph{Strategy 2: Random Keys Injected by Source and/or Intermediate
Nodes and Decoded at Sink}

In strategy 2, any node in the network can inject random keys.  The sink is required to decode both the secret message and the random keys from all nodes, i.e.~keys are not canceled within the network, while the random key rates must be sufficient to ``fill'' each wiretap set (in a sense that is made precise below). Although for simplicity of notation the algorithm description below is for the single-source, single-sink case, this strategy applies directly to multiple-source multicast case. If random keys are injected only at the source, the strategy reduces to the global key strategy in \cite{cai02}. Note that under the assumption that only the source knows the message and different
nodes do not have common randomness, here we cannot apply the key
cancelation and precoding idea from Strategy 1, since after applying the
precoding matrix each node may potentially be required to transmit a
mixture of the source message and other nodes' random keys.

Let $R_{w,v}$ be the random key injection
rate at node $v$. As before, $R_s$ denotes the secret message rate at the source and $z_{i,j}$ the transmission rate on link $(i,j)$.   We will address the choice of these rates below.  Consider the graph
$\mathcal{G}$ with the capacity of each link $(i,j)\in\links$ set as $z_{i,j}$. Connect each subset of links $\mathcal{A}\in \mathcal{W}$ to a virtual node $d^{\mathcal{A}}$: more precisely, for each directed link $(i,j)\in\links$ in the network, create a node $v_{i,j}$ and replace $(i,j)$ by two links $(i,v_{i,j})$ and $(v_{i,j},j)$ of capacity $z_{i,j}$, and for each $(i,j)\in\mathcal{A}$ create a link $(v_{i,j},d^{\mathcal{A}})$ of capacity $v_{i,j}$. Intuitively, we want the max flow/min cut capacity from the message and random key sources to $d^{\mathcal{A}}$ to be equal to that in the absence of the message. Similarly to strategy 1, we simplify the optimization by constraining this max flow/min cut capacity to be equal to an upper bound, $\sum_{(i,j)\in
\mathcal{A}} z_{i,j}$.  Specifically, we
have the following LP:
\begin{equation}\label{eq5_023}
\begin{split}
\max \text{ }&R_s\\
\text{subject to }& \sum_{j}f^{\mathcal{A}}_{i,j}-\sum_{j}f^{\mathcal{A}}_{j,i}\;\;\left\{%
\begin{array}{ll} =
  -\sum_{(i',j')\in
\mathcal{A}} z_{i',j'}, & \text{if }i=d^{\mathcal{A}}, \\
  %\kappa_i^{\mathcal{A}},
  \le R_{w,i},& \mathrm{otherwise},\\
\end{array}%
\right.\forall \mathcal{A}\in \mathcal{W},\\
& \sum_{j}f^d_{i,j}-\sum_{j}f^d_{j,i}=\left\{%
\begin{array}{ll}
  -\left(R_s+\sum_{v\in \mathcal{V},v\neq d} R_{w,v}\right), & \text{if }i=d, \\
  R_s+R_{w,s}, & \text{if }i=s, \\
  R_{w,i}, & \mathrm{otherwise},\\
\end{array}%
\right.\\
%& \sum_{v\in \mathcal{V}} \kappa_v^{\mathcal{A}}= \sum_{(i,j)\in
%\mathcal{A}} z_{i,j},
%\,\forall \mathcal{A}\in \mathcal{W},\\
%& \kappa_i^{\mathcal{A}}\leq R_{w,i},\quad
& f^{\mathcal{A}}_{i,j}
\leq z_{i,j},\quad f^{d}_{i,j}
\leq z_{i,j},\quad z_{i,j}\leq c_{i,j},\, \forall (i,j)\in
\mathcal{E},
\end{split}
\end{equation}
where the first set of equations corresponds to the requirement that %the max flow from the random key sources to $d^{\mathcal{A}}$ equals $\sum_{(i,j)\in\mathcal{A}} z_{i,j}$ (i.e.~
the network accommodates a flow $f^{\mathcal{A}}$ of size $\sum_{(i,j)\in
\mathcal{A}} z_{i,j}$ from the random key sources to $d^{\mathcal{A}}$, the second set of equations corresponds to the requirement that the network accommodates a flow $f^{d}$, of size equal to the sum of the message and random key rates, from the message and random key sources to the sink $d$,
%the max flow from the message and random key sources to the sink $d$ is equal to the sum of the message and random key rates,
and the third set of inequalities corresponds to the link capacity constraints.

\begin{claim}Strategy 2 allows the sink to decode the message $\mathbf{v}$, and achieves perfect secrecy.\end{claim}

\begin{proof}As illustrated in the example of Fig.~\ref{f3_2}, consider an augmented network with
\begin{itemize}
\item
a virtual source node $u_s $ connected to the
source node $s $ by a directed link $(u_s, s) $ of capacity $R_s$, and
connected to each virtual sink $ d^{\mathcal{A}}$ by a directed link $(u_s , d^{\mathcal{A}})$ of
capacity $ R_s$, and\item a virtual node $u_k$ connected to each node $v$ by
a directed link $(v,u_k) $ of capacity $R_{w,v}$, and connected to each
virtual sink $ d^{\mathcal{A}}$ by a directed link $(u_k , d^{\mathcal{A}})$
of capacity $\sum_v R_{w,v} - \sum_{(i,j)\in\mathcal{A}}
z_{i,j}$.\end{itemize}

\begin{figure}
\centering
  \includegraphics[scale=0.7]{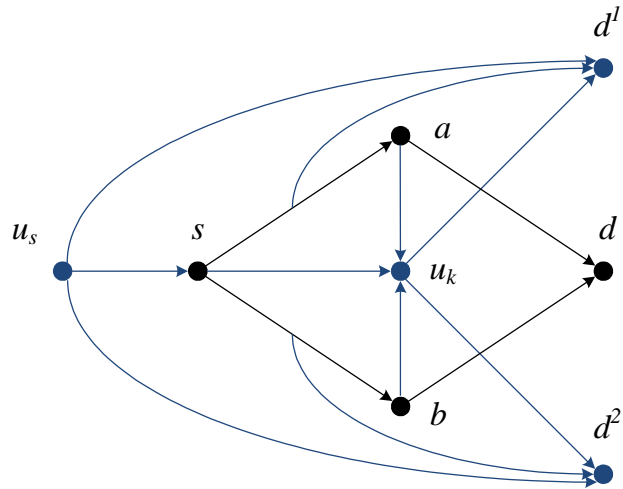}\\
  \caption{An example of the augmented network construction for the proof of
    correctness of strategy 2, where $s,a,b,d$ are nodes of the original graph, and only one of the two links  $(s,a)$ and $(s,b)$ can be wiretapped.}\label{f3_2}
\end{figure}

The source information enters the network at the
virtual source node $u_s$ and is transmitted to each virtual sink $
d^{\mathcal{A}}$. Consider a multi-source multicast problem on this network, where the
actual sink node and the virtual sinks $ d^{\mathcal{A}}$ each demand the source
message and all the random keys. By the first constraint of the LP, the max flow from the random key sources to $d^{\mathcal{A}}$ in the original network equals $\sum_{(i,j)\in
\mathcal{A}} z_{i,j}$; together with the additional capacity in the augmented network ($\sum_v R_{w,v} - \sum_{(i,j)\in\mathcal{A}}
z_{i,j}$ from the random key sources and $R_s$ from $u_s$), the
max flow from the message and random key sources to each virtual sink $d^{\mathcal{A}}$ is sufficient to ensure that the multicast problem is
feasible~\cite{ahlswede00}. A capacity-achieving code for this
multicast problem in the transformed graph corresponds to a code for the original secrecy problem,
since the information received by each virtual sink $ d^{\mathcal{A}}$ from
the set $\mathcal{A}$ of original network links must be independent of information received from the
additional links, which includes the entire source message. \end{proof}

\begin{figure}
\centering
  \includegraphics[scale=0.9]{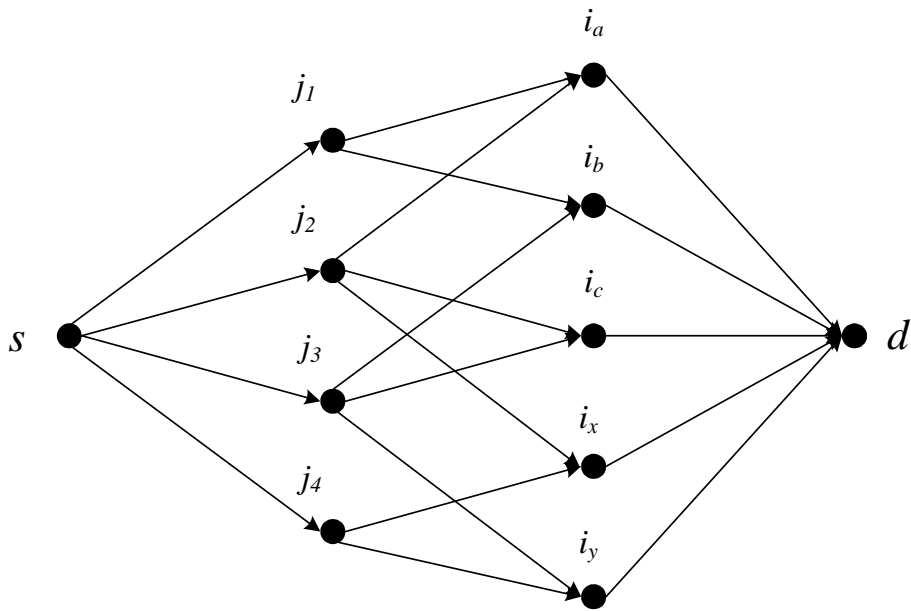}\\
  \caption{Example of the usefulness of Strategy 2.}\label{f31}
\end{figure}

An example where this strategy is useful is given in Fig.~\ref{f31},
which is obtained by interchanging the source and the sink as well
as reversing all the links in Fig.~\ref{f3}. At most three links in
the last layer can be wiretapped. By injecting one local key at node
$j_2$ and two global keys at the source, Strategy 2 can achieve
secrecy rate $2$. On the other hand, if random keys are only
injected at the source, the secrecy rate is at most $\frac{8}{5}$.
Let $R_s$ and $R_w$ be the secrecy rate and the random key rate at the
source, respectively. Let $z$ be the total rate of transmission on the last
layer. To achieve secrecy, we must have $R_w\geq \frac{3}{5}z$, where
the min-cut condition on the last layer requires $R_s+R_w\leq z$. Since
the source injects all the random keys, the min-cut condition on the
first layer requires $R_s+R_w\leq 4$. Combining these we obtain $R_s\leq
\frac{8}{5}$, which is strictly less than 2.

From the examples, we see that the types of scenarios in which Strategy 1 and Strategy 2 are useful seem to be complementary. In general, these two strategies can be combined to
obtain a higher secrecy rate. We use these strategies conceptually in the following sections to develop theoretical results. However, for numerical computation of achievable rates in scenarios 1 and 2, we note that the number of possible wiretapping sets, and thus the size of the LPs, are
exponential in the size $k$ of each wiretap set, so they are useful for small $k$.

\section{Unachievability of Cut Set Bound}\label{section:unachievable}

%From Theorem \ref{th31}, the secrecy cut-set bound is equal to the
%network  capacity where the wiretap links are erased.
In the case of
unrestricted wiretapping sets and unit link capacities, the secrecy
capacity is equal to the cut-set bound \cite{cai02}. In this section we
show that the cut-set bound (\ref{eq5_003}) is not achievable in general, by considering
the example in Fig.~\ref{f4}, where the set of wiretappable links is
restricted (Scenario 1). We give an explicit proof that the cut set
bound is not achievable for the case when the wiretap set is unknown. We also use the program  Information
Theoretic Inequalities Prover (Xitip) \cite{xitip} to show that the
secrecy capacity is bounded away from the cut set bound. We then
convert the example into one with unequal link capacities (Scenario
2), and show the unachievability of the cut set bound for this case
also.

\subsection{Restricted Wiretap Set (Scenario 1)}
\begin{figure}
\centering
  \includegraphics[scale=0.75]{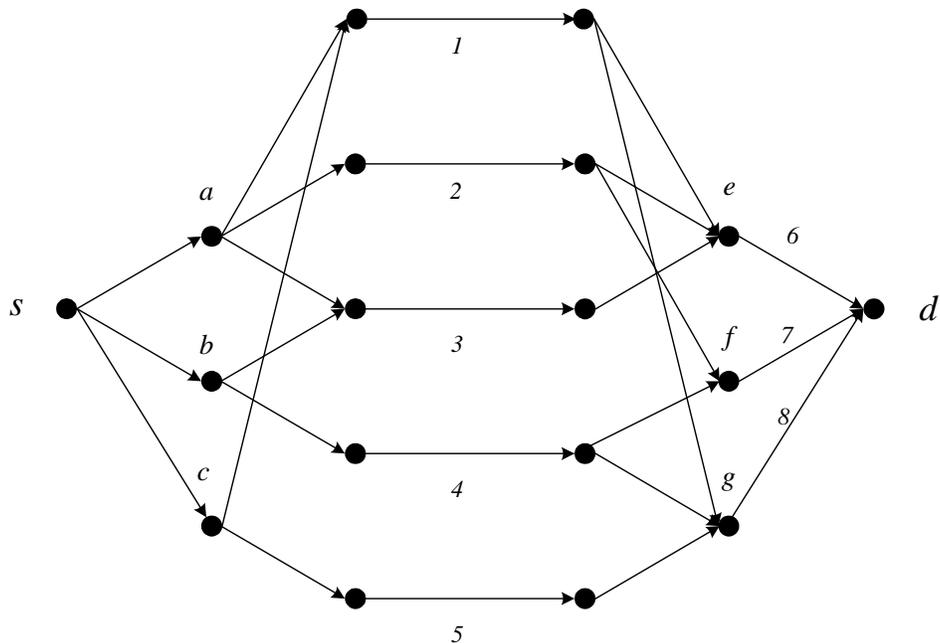} \\
  \caption{An example to show that the secrecy rate without knowledge of wiretapping set is smaller than that with such knowledge.
  The wiretapper can wiretap any three of the five links in the middle layer.}\label{f4}
\end{figure}

Consider the example in Fig.~\ref{f4}, where all links have unit capacity and any three of the five middle layer links can be wiretapped. Let the middle layer links be 1-5 (from top to bottom) and the last
layer links be 6-8 (from top to bottom).  Let the signal carried by link $i$ be called signal $i$,
or $S_i$. Let the source information be denoted $X$. The cut-set bound, or the secrecy capacity with known  wiretap set, is 2.

%The constraints required are that the source information is a
%function of the signals on the sink's incoming links, and that there
%is zero mutual information between the source information and the
%signals on the links in each adversarial subset.

%In this example, the cut set bound is 2.
To provide intuition for the case when the wiretap set is unknown, we
first show that secrecy rate 2 cannot be achieved by using scalar linear
coding. Then, the argument is converted to an information theoretic
proof that secrecy rate 2 cannot be achieved by using any possible coding
scheme.

Suppose secrecy rate 2 is achievable with a scalar linear network code.
First note that the source cannot inject more than unit amount of
random key, otherwise the first layer cannot carry two units of
source data. Let the random key injected by the source be denoted
$K$. For the case when the source injects a unit amount of random key, we first have the following observations. Signal 6 must be a
function of signal 1, otherwise if the adversary sees the signals
2-4 then he knows signals 6-7.  Also, signal 8 must be a function of
signal 5, otherwise if the adversary sees signals 1, 2 and 4, then
he knows signals 7-8. Similarly we can show that signal 8 must be a
function of signal 1, and signal 7 must be a function of signal 2.
We consider the following two cases.

Case 1: signal 5 is a linear combination of signals present at the
source node.  To achieve the full key rank condition on links 1, 2
and 5, node a must put two independent local keys
$k_1$ and $k_2$ on links 1 and 2 respectively.  Link 7, whose other
input is independent of $k_2$, is then a function of $k_2$.
Similarly, Link 8 is a function of $k_1$.  This means that the last
layer has two independent local keys on it.

Case 2: signal 5 is a linear combination of signals present at the
source node as well as a local key $k$ injected by node c.

Case 2a: $k$ is also present in signal 1.  Then $k$ is present in
signal 6, and is independent of the key present in signal 7.

Case 2b: $k$ is not present in signal 1.  Then $k$ is present in
signal 8, and is independent of the key present in signal 7.

In all three cases 1, 2a, and 2b, there is a pair of last layer links which
are functions of two independent random keys, leaving capacity for only one
unit of secret message. Thus, we conclude that the secrecy rate without
knowledge of the wiretapping set by using only linear network coding is less
than two.

We now extend the above argument to any
coding scheme which  leads to the following theorem.

\begin{theorem}
For the wireline network in Fig.~\ref{f4} a secrecy rate of 2 is not achievable with
any possible coding scheme, if any three out of the five links (1-5) in the
middle layer are wiretapped and the location of those links is unknown.
\end{theorem}

%Next, we convert this argument into an information theoretic proof
%that secrecy rate 2 is not achievable with any coding scheme.
\begin{proof}
See Appendix.
\end{proof}

We can also show that the secrecy rate is bounded away from 2 by
using the framework for linear information inequalities
\cite{yeung97}. Let $X$ be the message sent from the source and
$Z_i$, $i=1,2,3$ be the signals on the links adjacent to the
source. We want to check whether $H(X) \leq \omega$ is implied by
\begin{equation}\label{eq1000}
\begin{split}
(1)\quad& H(Z_i)\leq1,\,H(S_j)\leq1,\,i=1,2,3,\,j=1,\ldots,8,\\
(2)\quad& H(X|S_6,S_7,S_8)=0,\\
(3)\quad& I(X,Z_1,Z_2,Z_3,S_4,S_5,S_7,S_8;S_6|S_1,S_2,S_3)=0,\\
(4)\quad& I(X,Z_1,Z_2,Z_3,S_1,S_3,S_5,S_6,S_8;S_7|S_2,S_4)=0,\\
(5)\quad& I(X,Z_1,Z_2,Z_3,S_2,S_3,S_6,S_7;S_8|S_1,S_4,S_5)=0,\\
(6)\quad& I(X;S_1,S_2,S_3)=0,\, I(X;S_1,S_2,S_4)=0,\\
(7)\quad& I(X;S_1,S_2,S_5)=0,\, I(X;S_1,S_3,S_4)=0,\\
(8)\quad& I(X;S_1,S_3,S_5)=0,\, I(X;S_1,S_4,S_5)=0,\\
(9)\quad& I(X;S_2,S_3,S_4)=0,\, I(X;S_2,S_3,S_5)=0,\\
(10)\quad& I(X;S_2,S_4,S_5)=0,\, I(X;S_3,S_4,S_5)=0,\\
(11)\quad& I(S_1;Z_2|Z_1,Z_3)=0,\, I(S_2;Z_2,Z_3|Z_1)=0,\\
(12)\quad& I(S_3;Z_3|Z_1,Z_2)=0,\, I(S_4;Z_1,Z_3|Z_2)=0,\\
(13)\quad& I(S_5;Z_1,Z_2|Z_3)=0,\, I(S_1;S_4|Z_1,Z_2,Z_3)=0,\\
(14)\quad& I(S_2;S_4,S_5|Z_1,Z_2,Z_3)=0,\, I(S_3;S_5|Z_1,Z_2,Z_3)=0,\\
(15)\quad& I(S_4;S_1,S_2,S_5|Z_1,Z_2,Z_3)=0,\, I(S_5;S_2,S_3,S_4|Z_1,Z_2,Z_3)=0,\\
(16)\quad& I(S_1,S_2,S_3,S_4,S_5;X|Z_1,Z_2,Z_3)=0,
\end{split}
\end{equation}
where the first inequality is the capacity constraint, the second
constraint shows that the sink can decode $X$, constraints (3) to (5)
mean that the signals in the last layer are independent of other
signals given the incoming signals from the middle layer, constraints
(6) to (10) represent the secrecy constraints when any three links in
the middle layer are wiretapped, and constraints (11) to (16)
represent the conditional independence between the signals in the
first layer and those in the middle layer. In particular, (16) shows
that $X\rightarrow (Z_1,Z_2,Z_3)\rightarrow (S_1,\dots,S_5)$ forms a
Markov chain.  Note that constraints (3) to (5) and (11) to (16)
implicitly allow some randomness to be injected at the corresponding
nodes. We use the Xitip program \cite{xitip}, which relies on the
framework in \cite{yeung97}, to show that $H(X)\leq5/3$ is implied by
the set of equalities (\ref{eq1000}).  Therefore, $5/3$ is an upper
bound on the secrecy rate when the location of wiretapper is unknown,
which is less than the secrecy rate 2 achievable when such information
is known. Therefore, there is a strict gap between the secrecy
capacity and the cut set bound.

\begin{figure}
\centering
  \includegraphics[scale=0.75]{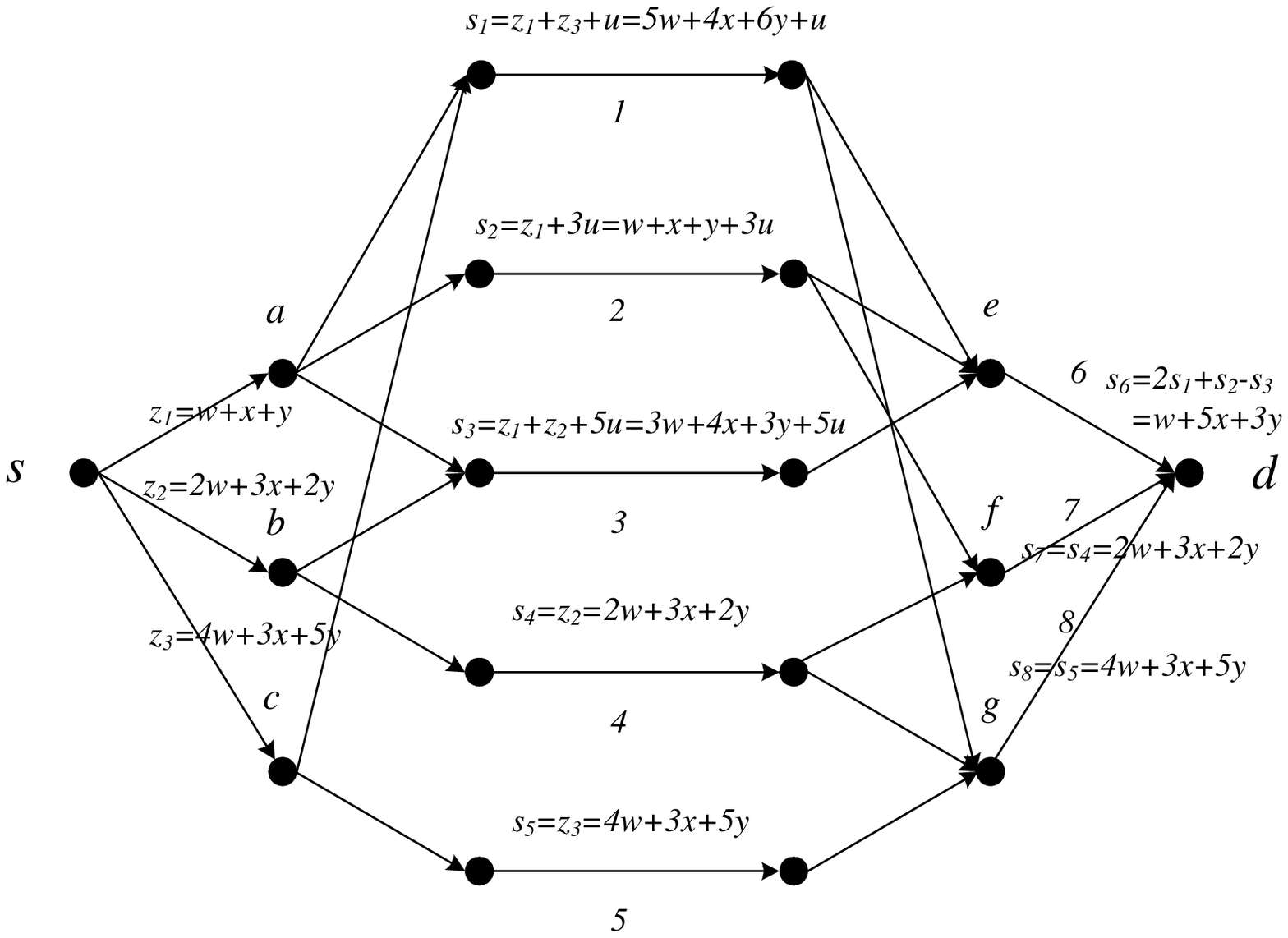} \\
  \caption{A coding scheme achieving secrecy rate 1 without knowledge of the
    wiretap set for the network in Fig.~\ref{f4}, where any three of the five
    middle layer links can be wiretapped. $w$ is the secret message, $x$ and
    $y$ are keys injected at the source, and $u$ is a key injected at node $a$
    and canceled at node $e$. The operations are over a finite field $GF(7)$.
  }\label{f4_1}
\end{figure}

On the other hand, the secrecy rate for the wireline network in Fig.~\ref{f4} is at least 1 which is shown by the example in Fig.~\ref{f4_1}, where a finite field $GF(7)$ is used.
%We can verify that $w$ can be received by the sink with perfect secrecy.
In this example, a combination of strategies 1 and 2 is used, where keys are
injected inside the network and are also canceled at intermediate nodes.

\subsection{Unequal Link Capacities (Scenario 2)}

We have restricted the wiretapped links to be in the middle layer in
Fig.~\ref{f4}. We next show that the unachievability of the cut-set
bound also holds for the secure network coding problem with unequal
link capacities (Scenario 2). We convert the example of
Fig.~\ref{f4} by partitioning each non-middle layer link into
$\frac{1}{\epsilon}$ parallel small links each of which has capacity
$\epsilon$. Any three links can be wiretapped in the transformed
graph. We
prove the unachievability of the cut-set bound in the transformed
network.

First, we show a lower bound on the min-cut between the source
and the sink in the transformed network when three links are deleted. Note that deleting any $k^\prime$ ($k^\prime\leq 3$) non-middle layer links reduces the min-cut by at most
$k^\prime\epsilon$. When $k^\prime = 0$, the min-cut is 2. When $k^\prime= 1$ or at most two
middle layer links are deleted, the min-cut  is at least 2 after deleting these middle layer links, and the min-cut is at least $2-k^\prime\epsilon\geq 2 -\epsilon$ after further deleting the $k^\prime= 1$ non-middle layer link. When $k^\prime= 2$ or at most one
middle layer link is deleted, the min-cut between the source
and the sink is 3 after deleting this middle layer link, and the min-cut is at least $3-k^\prime\epsilon\geq 3 -3\epsilon$ after further deleting the $k^\prime$ non-middle layer links.
Therefore, the cut-set bound is at least $\min (2 -\epsilon, 3 -3\epsilon)$.

For the case where the location of the wiretap links is unknown, we
prove the unachievability of the cut-set bound in the transformed
network. First, consider the transformed network with the
restriction that the wiretapper can only wiretap any 3 links in the
middle layer. The optimal solution is exactly the same as for the
original network of the previous subsection, and achieves secrecy
rate at most $5/3$. Now, consider the transformed network without
the restriction on wiretapping set, i.e., the wiretapper can wiretap
any 3 links in the entire network. As wiretapping only the middle
layer links is a subset of all possible strategies that the
wiretapper can have, the secrecy rate in the transformed network is
less than or equal to that in the former case, which is strictly
smaller than the cut-set bound for $\epsilon$ strictly smaller than $\frac{1}{4}$. Therefore, the cut-set bound is
still unachievable when the wiretap links are unrestricted in the
transformed graph.

\section{NP-hardness}\label{section:NP}

We show in the following that determining the secrecy capacity %in Scenarios 1 and 2 (defined in Section~\ref{sect5_2})
is NP-hard %when the location of the wiretap links is known or unknown
by reduction from the clique problem, which determines whether a
graph contains a clique\footnote{A
  clique in a graph is a set of size $r$ of pairwise adjacent vertices, or in
  other words, an induced subgraph which is a complete graph.} of at
least a given size $r$.

When the choice of the wiretap set is made known to the communicating nodes, the secrecy capacity  is given by the cut-set bound, from Theorem~\ref{th31}, and is achieved by not transmitting on the wiretapped links. Finding the cut-set bound involves determining the worst case wiretap set. This is equivalent
to the network interdiction problem
\cite{wood93}, which is to minimize the
maximum flow of the network when a given number of links in the
network are removed.
It is shown in \cite{wood93} that the network interdiction problem is
NP-hard. Therefore, determining the secrecy capacity for the case where the location of the wiretap links
is known is NP-hard.

To show that determining the secrecy capacity for the case where the location of
the wiretap links is unknown is NP-hard, we use the construction in
\cite{wood93} showing that for any clique problem on a given graph
$\mathcal{H}$, there exists a corresponding network
$\mathcal{G}^{\mathcal{H}}$ whose secrecy capacity is $r$ when the
location of the wiretap links is known if and only if $\mathcal{H}$
contains a clique of size $r$. We then show that for all such networks $\mathcal{G}^{\mathcal{H}}$,
the secrecy capacity for the case when the location of the wiretap links
is unknown is equal to that for the case when this information is
known, which shows that there is a one-to-one correspondence between the
clique problem and the secrecy capacity problem.

\begin{figure}
\centering
  \subfigure[Original Graph $\mathcal{H}$]{\includegraphics[scale=0.7]{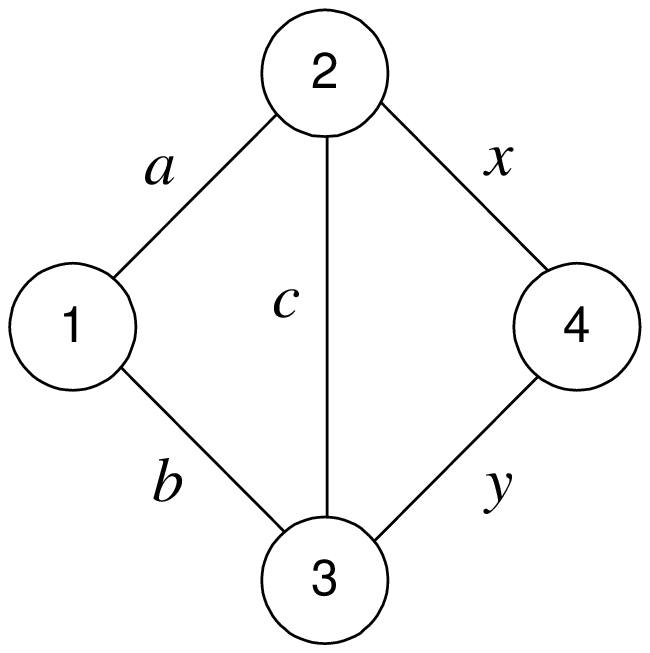}}
  \subfigure[Transformed Graph $\mathcal{G}^{\mathcal{H}}$]{\includegraphics[scale=0.7]{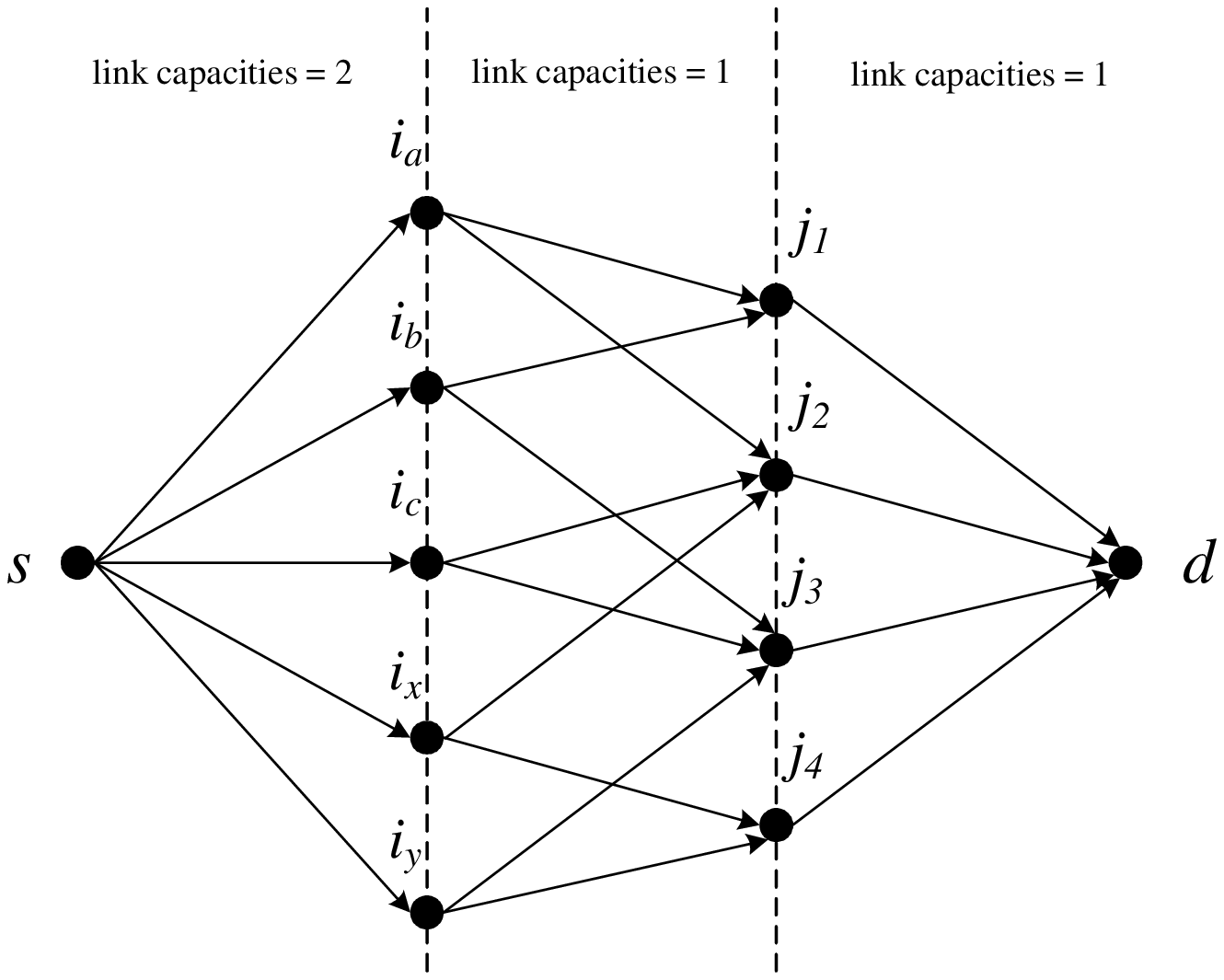}}\\
  \caption{Example of NP-hardness proof for the case with
    knowledge of the wiretapping set.}
  \label{f2}
\end{figure}

We briefly describe the approach in \cite{wood93} in the following.
Given an undirected graph
$\mathcal{H}=(\mathcal{V}_h,\mathcal{E}_h)$, we will define a
capacitated directed network $\hat{\mathcal{G}}^{\mathcal{H}}$ such
that there exists a set of links $\hat{\mathcal{A}}^\prime$ in
$\hat{\mathcal{G}}^{\mathcal{H}}$ containing less than or equal to
$|\mathcal{E}_h|-{r\choose 2}$ links such that
$\hat{\mathcal{G}}^{\mathcal{H}}-\hat{\mathcal{A}}^\prime$ has a
maximum flow of $r$ if and only if $\mathcal{H}$ contains a clique
of size $r$. For a given undirected graph
$\mathcal{H}=(\mathcal{V}_h,\mathcal{E}_h)$ without parallel links
and self loops, we create a capacitated, directed graph
$\mathcal{G}^{\mathcal{H}}=(\mathcal{N},\mathcal{A})$ as follows:
For each link $e\in \mathcal{E}_h$ create a node $i_e$ in a node set
$\mathcal{N}_1$ and for each vertex $v\in \mathcal{V}_h$ create a
node $j_v$ in a node set $\mathcal{N}_2$. In addition, create source
node $s$ and destination node $d$. For each link $e\in
\mathcal{E}_h$, direct a link in $\mathcal{G}^{\mathcal{H}}$ from
$s$ to $i_e$ with capacity 2 and call this set of links
$\mathcal{A}_1$. For each link $e=(u,v)\in \mathcal{E}_h$, direct
two links in $\mathcal{G}^{\mathcal{H}}$ from $i_e$ to $j_v$ and
$j_u$ with capacity 1, respectively and call this set of links
$\mathcal{A}_2$. For each vertex $v\in \mathcal{V}_h$, direct a link
with capacity 1 from $j_v$ to $d$. Let this be the set of links
$\mathcal{A}_3$. This completes the construction of
$\mathcal{G}^{\mathcal{H}}=(\mathcal{N},\mathcal{A})=(\{s\}\cup\{d\}\cup
\mathcal{N}_1\cup\mathcal{N}_2,\mathcal{A}_1\cup\mathcal{A}_2\cup\mathcal{A}_3)$.
In Fig.~\ref{f2}, we give an example of the graph transformation,
where $\mathcal{H}=(\{1,2,3,4\},\{a,b,c,x,y\})$. We use the following result from \cite{wood93}:

\begin{lemma}[{\cite[Lemma 2]{wood93}}] \label{lm31}
Let $\mathcal{G}^{\mathcal{H}}$ be constructed from $\mathcal{H}$ as
above. Then, there exists a set of links
$\mathcal{A}_1^\prime\subseteq \mathcal{A}_1$ with
$|\mathcal{A}_1^\prime|=|\mathcal{E}_h|-{r\choose 2}$ such that the
maximum flow from $s$ to $d$ in
$\mathcal{G}^{\mathcal{H}}-\mathcal{A}_1^\prime$ is $r$ if and only
if $\mathcal{H}$ contains a clique of size $r$.
\endproof
\end{lemma}

After obtaining $\mathcal{G}^{\mathcal{H}}$, we generate
$\hat{\mathcal{G}}^{\mathcal{H}}$ by replacing each link $(i_e,j_v)$
with $|\mathcal{E}_h|$ parallel links each with capacity
$1/|\mathcal{E}_h|$ and call this link set $\hat{\mathcal{A}}_2$. We
carry out the same procedure for links $(j_v,d)$ and call this link
set $\hat{\mathcal{A}}_3$. Then
$\hat{\mathcal{G}}^{\mathcal{H}}=(\mathcal{N},\mathcal{A})=(\{s\}\cup\{d\}\cup
\mathcal{N}_1\cup\mathcal{N}_2,\mathcal{A}_1\cup\hat{\mathcal{A}}_2\cup\hat{\mathcal{A}}_3)$.
For the case when the location of wiretap links is known, it is
shown in \cite{wood93} that the worst case wiretapping set
$\hat{\mathcal{A}}^\prime$ must be a subset of $\mathcal{A}_1$. By
using Lemma \ref{lm31}, this case is NP-hard.

%Now, we consider the case where the wiretapping set is unknown, and
%show that the secrecy capacity of $\hat{\mathcal{G}}^{\mathcal{H}}$
%when the wiretapper accesses any unknown subset of
%$k=|\mathcal{E}_h|-{r\choose 2}$ links is at least $r$ if and only if
%$\mathcal{H}$ contains a clique of size $r$.
Now, we consider the secrecy capacity when $k=|\mathcal{E}_h|-{r\choose 2}$ and the wiretapping set is unknown. From Lemma \ref{lm31},
the condition that $\mathcal{H}$ contains a clique of size $r$ is
equivalent to the condition that the max-flow to the sink in
$\mathcal{G}^{\mathcal{H}}$ after removing any $k$ links from
$\mathcal{A}_1$ is at least $r$. We now show that the latter condition is
equivalent to the condition that the secrecy capacity of
$\mathcal{G}^{\mathcal{H}}$ when the wiretapper accesses any unknown
subset of $k$ links from $\mathcal{A}_1$ (Scenario 1) is at least $r$. For each subset $\mathcal{A}^\prime$ of $k$ links from $\mathcal{A}_1$, we create nodes
$t^{\mathcal{A}^\prime_1}$ and
$d^{\mathcal{A}^\prime_1}$ with their corresponding incident links as described in Strategy 1.  As the wiretapped links each have capacity 2 and are connected to the
source directly, the min-cut between the source and each virtual sink $d^{\mathcal{A}^\prime_1}$ is at least $2k+r$.
%Since $r$ is the cut-set upper bound on the secrecy rate, by
Then, by using Strategy 1 the secrecy rate $r$ is
achievable. %, which is equal to the secrecy rate when the location of wiretap links is known.

Finally, we show that the same condition is also equivalent to the condition that the
%secrecy capacity of
%$\mathcal{G}^{\mathcal{H}}$ when any $k$ links of $\mathcal{A}_1$
%are wiretapped (Scenario 1, discussed in the previous paragraph) is equal to the
secrecy capacity of
$\hat{\mathcal{G}}^{\mathcal{H}}$ when any $k$ links are wiretapped
(Scenario 2) is at least $r$. Since each second layer link has a single first layer
link as its only input, wiretapping a second layer link yields no
more information to the wiretapper than wiretapping a first layer
link. When some links in the third layer are wiretapped, let the
wiretapping set be
$\hat{\mathcal{A}}^\prime=\hat{\mathcal{A}}^\prime_1\cup\hat{\mathcal{A}}^\prime_3$
where %$\hat{\mathcal{A}}^\prime_3\neq\emptyset$ or
$|\hat{\mathcal{A}}^\prime_3|\geq 1$ and
$|\hat{\mathcal{A}}^\prime_1|\leq k-1$. Thus
$\mathcal{A}_1-\hat{\mathcal{A}}^\prime_1$ contains at least
${r\choose 2}+1$ links. We create nodes
$t^{\hat{\mathcal{A}}^\prime}$ and
$d^{\hat{\mathcal{A}}^\prime}$ with their corresponding incident links as described in Strategy 1. Since removing links in $\mathcal{A}_1$
corresponds to removing links in $\mathcal{H}$, after removing
links in $\mathcal{H}$ corresponding to
$\hat{\mathcal{A}}^\prime_1$, $\mathcal{H}$ contains a subgraph
$\mathcal{H}_1$ containing ${r\choose 2}$ links plus at least one link $e=(u,v)$. %We consider two cases.

Case 1: $\mathcal{H}_1$ is a clique of size $r$. In this case, the
number of vertices with degree greater than 0 in $\mathcal{H}_1\cup
e$ is $r+2$.

Case 2: $\mathcal{H}_1$ is not a clique. $\mathcal{H}_1$ contains at
least $r+1$ vertices with degree greater than 0.

According to \cite[Lemma 1]{wood93}, the max-flow in
$\mathcal{G}^{\mathcal{H}}$ is equal to the number of vertices in
$\mathcal{H}$ with degree greater than 0. In both cases, the
max-flow of  $\mathcal{G}^{\mathcal{H}}$ after removing links in
$\hat{\mathcal{A}}^\prime_1$ is at least $r+1$. Let $\tilde{R}_{s\rightarrow
\hat{\mathcal{A}}_3^\prime}$ be the max-flow capacity from the
source to $\hat{\mathcal{A}}_3^\prime$ in
$\hat{\mathcal{G}}^{\mathcal{H}}-\hat{\mathcal{A}}^\prime_1$.

We can use a variant of the Ford-Fulkerson (augmenting paths) algorithm, e.g., \cite{dean06finite},
as follows to construct a max-flow subgraph
$\mathcal{D}$ from $s$ to $\hat{\mathcal{A}}_3^\prime$ in
$\hat{\mathcal{G}}^{\mathcal{H}}-\hat{\mathcal{A}}^\prime_1$ satisfying the property that after removing
$\mathcal{D}$ from
$\hat{\mathcal{G}}^{\mathcal{H}}-\hat{\mathcal{A}}^\prime_1$, the min-cut
between $s$ and $d$ is at least
\begin{eqnarray}\nonumber r+1-\tilde{R}_{s\rightarrow
\hat{\mathcal{A}}_3^\prime}&\ge &r+1-|\hat{\mathcal{A}}^\prime_3|/|\mathcal{E}_h|\\\nonumber &\ge & r+1-(|\mathcal{E}_h|-1)/|\mathcal{E}_h|\\&>&r,\label{equation:paths}\end{eqnarray}where we have used $|\hat{\mathcal{A}}^\prime_3| \le |\mathcal{E}_h|-1$. Considering the network $\hat{\mathcal{G}}^{\mathcal{H}}-\hat{\mathcal{A}}^\prime_1$ with all link directions reversed, we construct augmenting paths via depth first search from $d$ to $s$, starting first by constructing augmenting paths
via  links in $\hat{\mathcal{A}}_3^\prime$, until we obtain a set of paths corresponding to a max flow of capacity $\tilde{R}_{s\rightarrow
\hat{\mathcal{A}}_3^\prime}$ between $s$ and $\hat{\mathcal{A}}_3^\prime$.
We add further augmenting paths until we obtain a max flow (of capacity at least $r+1$) between $s$ and $d$, which may cause some of the paths traversing links in $\hat{\mathcal{A}}_3^\prime$ to be redefined but without changing their  total capacity. The subgraph $\mathcal{D}$ consists of the final set of paths traversing links in $\hat{\mathcal{A}}_3^\prime$. Thus, the paths remaining after removing $\mathcal{D}$ have a total capacity lower bounded by (\ref{equation:paths}).

Therefore, the min-cut between the source and
$d^{\hat{\mathcal{A}}^\prime}$ in
$\hat{\mathcal{G}}^{\mathcal{H}}-\hat{\mathcal{A}}^\prime_1-\mathcal{D}$
is at least $r$, and the min-cut between the source and
$d^{\hat{\mathcal{A}}^\prime}$ in $\hat{\mathcal{G}}^{\mathcal{H}}$ is at least
$r+R_{s\rightarrow \hat{\mathcal{A}}^\prime_1}+\tilde{R}_{s\rightarrow
\hat{\mathcal{A}}_3^\prime}=r+R_{s\rightarrow \hat{\mathcal{A}}^\prime}$.
By using Strategy 1, a secure rate of $r$ is achievable when $\hat{\mathcal{A}}^\prime$ is wiretapped. Thus, the secrecy rate
for the case when the location of the wiretap links is unknown is
equal to that for the case when such information is known with
an unrestricted wiretapping set. We have thus proved the following theorem.

\begin{theorem}\label{th36}
  For a single-source single-sink network consisting of
  point-to-point links and an unknown wiretapping set, computing the
  secrecy capacity is NP-hard.
\end{theorem}

\section{Conclusion}\label{sect5_6}
In this paper, we addressed the secrecy capacity of wireline
networks where different links have different capacities. In
particular, it was shown that the secrecy capacity is not the same
in general when the location of the wiretapped links is known or
unknown; in the former case the capacity is given by a cut-set bound, which is unachievable in general in the latter case.  Further, we proposed achievable strategies where random
keys are canceled at intermediate non-sink nodes, or injected at
intermediate non-source nodes. Finally, we showed that determining
the secrecy capacity is an NP-hard problem.% no matter whether the
%location of the wiretapped links is known or unknown.

\section*{Appendix: Proof of Theorem 2}
\label{sec:proofth2}
We prove Theorem 2 by contradiction. Suppose that a secrecy rate of 2 is
achievable for the network in Fig.~\ref{f4}. As before, let $X$ and $K$ denote
respectively the secret message and random key injected by the source node,
and $S_i$ the signal on link $i$. Then each triple of links in the middle
layer has zero mutual information with the source data, and each pair of links
in the middle layer has joint conditional entropy 2 given the other three
links.

Since the message $X$ is decodable from information on the last layer, we have
$I(S_6,S_7,S_8;X)=2$.  %$I(S_6,S_7,S_8;K)\leq 1$.
Since $I(S_1,S_2,S_3;X)=0$, by the data processing inequality $I(S_6;X)=0$,
therefore, $I(S_7,S_8;X|S_6) = 2$ and $H(S_7|S_6)= I(S_7;X|S_6) = 1$. Then,
$H(S_7|X,S_6)=H(S_7|S_6) -I(S_7;X|S_6) = 0$. This implies that $S_7$ does not
depend on random keys injected by nodeS $f$ or $\mbox{head}(4)$ which would be
independent of $X,S_6$. Similarly, $I(S_8;X)=0$, implying $H(S_7|S_8)=
I(S_7;X|S_8) = 1$ and $H(S_7|X,S_8) = 0$. Thus, $S_7$ does not depend on
random keys injected by node $\mbox{head}(2)$ which would be independent of
$X$ and $S_6$. In a similar manner, we can show that $S_6$ and $S_8$ also do
not depend on any random keys injected after the middle layer.  Also, since
$H(S_7,S_8|S_6)\ge I(S_7,S_8;X|S_6) = 2$ and $H(S_6)\ge H(S_6|S_8 )= 1$,
therefore $H(S_6,S_7,S_8)=3$.  Let $S_A $ denote the adversary's observations.
By the secrecy requirement, $H(S_6,S_7,S_8|S_A)= 2$, which implies
$I(S_6,S_7,S_8;S_A) =H(S_6,S_7,S_8) -H(S_6,S_7,S_8|S_A)= 1$.

%, and no additional random keys
%can be injected after the middle layer. Then the adversary can know
%the signal on any one of the links in the last layer, so there must
%be one unit of random key on the last layer. Therefore, the coding
%scheme must ensure that the mutual information between the
%adversary's observations and the information on the last layer is 1.
Then, the mutual information $I(S_6;S_2,S_3)=0$, otherwise, if the adversary
sees signals 2-4 his mutual information with signals 6-7 is greater than 1.
The mutual information $I(S_8;S_1,S_4)=0$, otherwise if the adversary sees
signals 1, 2, 4 his mutual information with signals 7-8 is greater than 1.
The mutual information $I(S_8;S_4,S_5)=0$, otherwise if the adversary sees
signals 2, 4, 5 his mutual information with signals 7-8 is greater than 1.
The mutual information $I(S_7;S_4,S_5)=0$, otherwise if the adversary sees
signals 1, 4, 5, his mutual information with signals 7-8 is greater than 1.

Case 1: signal 5 is a function of only signals present at the source
node, i.e., $H(S_5|X,K)=0$. By the zero mutual information condition
for links 1, 2 and 5, $H(S_1,S_2,S_5|X)=3$, so
\begin{equation}
H(S_1,S_2,S_5|X,K)=H(S_1,S_2|X,K,S_5) =2.
\end{equation}
Since $S_4$ is conditionally independent of $S_1$, $S_2$ given $X$
and $K$, we have $H(S_1,S_2|X,K,S_4,S_5) =2$,
$I(S_1,S_2;X,K,S_4,S_5)=0$ and $I(S_1,S_2;X,K|S_4,S_5)=0$. Now
\begin{equation}
\begin{split}
I(S_1,S_2,S_7,S_8;X,K|S_4,S_5)=&I(S_7,S_8;X,K|S_4,S_5)+I(S_1,S_2;X,K|S_7,S_8,S_4,S_5)\\
=&I(S_1,S_2 ;X,K|S_4,S_5) +I(S_7,S_8;X,K|S_1,S_2,S_4,S_5).
\end{split}
\end{equation}
Since $S_7,S_8$ is conditionally independent of $X,K$ given
$S_1,S_2,S_4,S_5$, we have
\begin{equation}
I(S_7,S_8;X,K|S_1,S_2,S_4,S_5) = 0.
\end{equation}
Then by the non-negativity of conditional mutual information,
\begin{equation}
I(S_7,S_8;X,K|S_4,S_5) \leq I(S_1,S_2;X,K|S_4,S_5) = 0.
\end{equation}
Next, note that $S_1$ and $S_2$ are conditionally independent given
$S_4$ and $S_5$, since $H(S_1|S_4,S_5)=H(S_2|S_1,S_4,S_5) =1$.
Therefore $S_7$ and $S_8$ are conditionally independent given $S_4$
and $S_5$, i.e. $I(S_7;S_8|S_4,S_5)=0$. Since
$H(S_7|S_4,S_5)=H(S_7)-I(S_7;S_4,S_5)=1$, it follows that
$H(S_7|S_8,S_4,S_5)=1$.  Then we have
\begin{equation}
\begin{split}
I(S_7,S_8;S_4,S_5) =&I(S_8;S_4,S_5) +I(S_7;S_4,S_5|S_8)\\
=&I(S_8;S_4,S_5) +H(S_7|S_8)-H(S_7|S_4,S_5,S_8)=0+1-1 = 0.
\end{split}
\end{equation}
So,
$I(S_7,S_8;X,K,S_4,S_5)=I(S_7,S_8;X,K|S_4,S_5)+I(S_7,S_8;S_4,S_5) =
0$, and therefore $H(S_7,S_8|X)\geq H(S_7,S_8|X,K,S_4,S_5)=2$, which
contradicts the requirement that there is at most 1 unit of random key on the last layer.

Case 2: signal 5 is not a function only of signals present at the
source

Case 2a: signal 1 has nonzero mutual information with some random key injected at node $c$. Then $H(S_1|X,K,S_2,S_3,S_4) >0$. For
brevity, let $A=(S_2,S_3)$ and $Y=(X,K,S_4)$. Since $I(S_6;A) = 0$
and $H(S_6|S_1,A) = 0$, we have $H(A)+H(S_6)=H(A,S_6) \leq H(A,S_1)
=H(S_1)+H(A|S_1)$. Since $H(S_6)=H(S_1)$, we have $H(A)=H(A|S_1)$
and so $H(S_1|A)=H(S_1)$.  Then from
$H(S_1,S_6|A)=H(S_1|A,S_6)+H(S_6|A)=H(S_6|A,S_1)+H(S_1|A)$, we have
$H(S_1|A,S_6)=0$.  Since $H(S_1|A,Y,S_6) \leq H(S_1|A,S_6)= 0$ and
$H(S_6|A,Y,S_1) \leq H(S_6|A,S_1) = 0$, from
\begin{equation}
I(S_1;S_6| Y,A) =H(S_1|A,Y)
-H(S_1|A,Y,S_6)=H(S_6|A,Y)-H(S_6|A,Y,S_1)>0
\end{equation}
we have $H(S_6|A,Y) = H(S_1|A,Y)
>0$. Then since $H(S_7|S_2,S_4)=0$, we have $H(S_6|S_7,X)>0$.  Also, since
$H(S_7|X)=1$, we have $H(S_6,S_7|X)>1$.

Case 2b: signal 1 has zero mutual information with any random key
injected at node c. Then $H(S_5|X,K,S_1,S_2,S_4) >0$.  Similar
reasoning as for case 2a applies with $A=(S_1,S_4)$, $Y=(X,K,S_2)$,
$S_5$ in place of $S_1$, and $S_8$ in place of $S_6$.

From Cases 1, 2a, and 2b, we conclude that the secrecy rate without
knowledge of the wiretapping set by using any nonlinear or linear
coding strategy is smaller than two obtained for the case where such
knowledge is present at the source.

%\end{appendix}

\bibliographystyle{IEEEtran} % plain,unsrt,alpha,abbrev
\bibliography{IEEEabrv,ref}

\end{document}